\documentclass[onecolumn,12pt,lettersize]{IEEEtran}
\usepackage[top=1 in, bottom=1 in, left=0.75 in, right=0.75 in]{geometry}

\usepackage[cmex10]{amsmath} 
\usepackage{amssymb}
\usepackage{tabularx}
\usepackage[sc,osf,noBBpl]{mathpazo} 
\usepackage{cite} 
\usepackage[pdftex]{graphicx} 
\usepackage{subcaption}
\usepackage{algorithm}
\usepackage[noend]{algorithmic}
\usepackage{amssymb} 
\usepackage{graphicx}
\usepackage{epstopdf}

\usepackage{stfloats}
\usepackage{tikz}

\usepackage{amsthm} 
\newtheorem{theorem}{Theorem}
\newtheorem{lemma}{Lemma}

\newtheorem{cor}{Corollary}

\begin{document}

\title{On the Rate of Convergence of the Power-of-Two-Choices to its Mean-Field Limit}
\author{
Lei Ying, ECEE, Arizona State University
}

\maketitle
\begin{abstract}
This paper studies the rate of convergence of the power-of-two-choices, a celebrated randomized load balancing algorithm for many-server queueing systems, to its mean field limit. The convergence to the mean-field limit has been proved in the literature, but the rate of convergence remained to be an open problem. This paper establishes that the sequence of stationary distributions, indexed by $M,$ the number of servers in the system, converges in mean-square to its mean-field limit with rate $O\left(\frac{(\log M)^3 (\log\log M)^2}{M}\right).$
\end{abstract}

\section{Introduction}
The power-of-two-choices is a celebrated randomized load balancing algorithm for the {\em supermarket model}, under which each incoming job is routed to the shorter of two randomly sampled servers from $M$ servers. It has been shown in  \cite{Mit_96,VveDobKar_96} that in the infinite server regime (also called the mean-field limit), the power-of-two-choices reduces the mean queue length (at each server) from $\Theta\left(\frac{1}{1-\rho}\right)$ to $\Theta\left(\log \frac{1}{1-\rho}\right),$ where $\rho$ is the load of the system. Besides the convergence of the stationary distributions to its mean-field limit, and the process-level convergence over a finite time interval has been studied in \cite{LucNor_05}, the convergence analysis with general service-time distributions can be found in \cite{BraLuPra_12}, and the analysis with heterogeneous servers can be found in \cite{MukMaz_13}. This seminal work has also inspired new randomized load balancing algorithms in recent years (see, e.g., \cite{TsiXu_12,LuXieKli_11,YinSriKan_15,XieDonLu_15,MukKarMaz_15}).

The convergence proofs in \cite{Mit_96,VveDobKar_96} are based on the interchange of the limits, which include the following steps: (1) proving the convergence of the continuous-time Markov chain (CTMC) to the solution of a dynamical system (the mean-field model) over a finite time interval $[0,t]$ as $M\rightarrow\infty,$ i.e., \begin{equation*}\lim_{M\rightarrow \infty}\sup_{0\leq s\leq t}d({\bf x}^{(M)}(s),{\bf x}(s))=0,\end{equation*} where ${\bf x}^{(M)}$ is the CTMC, $\bf x$ is the solution of the dynamical system, and $d(\cdot,\cdot)$ is some measure of distance; (2) proving that the mean-field model converges to a unique equilibrium point starting from any initial condition, i.e., $$\lim_{t\rightarrow\infty} {\bf x}(t)={\bf x}^*,$$ where $\bf x^*$ is the equilibrium point;  and (3) establishing the convergence of the stationary distributions to its mean-field limit via the interchange of the limits, i.e., \begin{align*}
\lim_{M\rightarrow\infty}\lim_{t\rightarrow\infty} {\bf x}^{(M)}(t)=\lim_{t\rightarrow\infty} \lim_{M\rightarrow\infty}  {\bf x}^{(M)}(t) ={\bf x}^*.
\end{align*} This approach based on the interchange of the limits and process-level convergence is an {\em indirect} method for proving the convergence of stationary distributions, so does not answer the rate of convergence (i.e., the approximation error of using the mean-field limit for the finite-size system). Over last few years, a new approach has emerged for quantifying the rate of convergence of queueing systems to its mean-field or diffusion limit based on Stein's method. Gurvich \cite{Gur_14} used the method for steady-state approximations for exponentially ergodic Markovian queues. In \cite{BraDai_15} a modular framework has been developed by Braverman and Dai for steady-state diffusion approximations and has been used for quantifying the rate of convergence to diffusion models for $M/Ph/n+M$ queuing systems. An approach similar to Stein's method has also been used by Stolyar in \cite{Sto_15_2} to show the tightness of diffusion-scaled stationary distributions. An introduction to Stein's method for steady-state diffusion approximations can be found in \cite{BraDaiFen_15}. Based on Stein's method and the perturbation theory,  Ying recently established the rate of convergence of a class of stochastic systems that can be represented by {\em finite-dimensional} population processes \cite{Yin_16}.

This paper exploits the method in \cite{Yin_16}. However, the result in \cite{Yin_16} only applies to continuous-time Markov chains (CTMCs) with {\em finite} state space so that the CTMC can be represented as a finite-dimensional population process. The main contributions of this paper are two-fold. First, it establishes the rate of convergence of the power-of-two-choices under the super-market model, where the state space of the CTMC is countable instead of finite, and the corresponding population process is an infinite dimensional process. From the best of our knowledge, the rate of convergence of the power-of-two-choices to its mean-field limit has not been established. In fact, little is known for the rate of convergence of infinite-dimensional queueing systems (e.g., \cite{TsiXu_12,YinSriKan_15}) to their mean-field limits. Another contribution of this paper is that  it demonstrates that the approach based on Stein's method and the perturbation theory \cite{Yin_16} can be extended to infinite-dimensional systems. While general conditions such as those in \cite{Yin_16} appear difficult to establish, the approach itself can be applied for analyzing other infinite-dimensional systems. In particular, this paper considers a so-called {\em imperfect} mean-field model \cite{Yin_16}, which is a finite-dimensional, truncated version of the dynamical system corresponding to the infinite-dimensional CTMCs (the perfect mean-field model). Specifically, this paper constructs a truncated dynamical system with dimensional $n=\Theta(\log M)$. The equilibrium point of the truncated system is consistent with the first $n$-dimension of the mean-field limit. Then we prove the approximation error from the mean-field model is $O\left(\frac{n^3\left(\log n\right)^2}{M}\right)$ and the error from the truncation is $O\left(\rho^n n\log n\right)$. After establishing the two facts above, the rate of convergence in mean square is obtained by choosing $n=\Theta(\log M).$

\section{The power-of-two-choices and its mean-field limit}
This paper studies the super market model in \cite{Mit_96,VveDobKar_96} with $M$ identical servers, each maintaining a separate queue. Assume jobs arrive at the system following a Poisson process with rate $\lambda M$ and the processing time of each job is exponentially distributed with mean processing time $\mu=1.$ Let $Q_m(t)$ denote the queue size of server $m$ at time $t.$ For each incoming job, the router (or called a scheduler) routes the job using the power-of-two-choices algorithm, which randomly samples two servers and dispatches the job to the server with a smaller queue size. In this setting, ${\bf Q}(t)$ is a CTMC and has a unique stationary distribution when $\lambda <1$ \cite{VveDobKar_96}.

Let $s^{(M)}_k(t)$ denote the fraction of servers with queue size {\em at least} $k$ at time instant $t,$ where the superscript $M$ denotes the system size.  ${\bf s}^{(M)}(t)$ is also a CTMC with the following transition rates:
\begin{equation}
Q_{{\bf s}, {\bf s}'}=\left\{
                        \begin{array}{ll}
                          M(s_k-s_{k+1}), &  \hbox{ if } {\bf s}'={\bf s}-\frac{{\bf 1}_k}{M} \\
                          \lambda M\left(s_{k-1}^2-s_k^2\right), & \hbox{ if } {\bf s}'={\bf s}+\frac{{\bf 1}_k}{M}\\
                          \sum_{k=1}^{\infty}-\lambda M\left(s_{k-1}^2-s_k^2\right)-M(s_k-s_{k+1}), & \hbox{ if } {\bf s}'={\bf s}\\
                          0, & \hbox{otherwise.}
                        \end{array}
                      \right., \label{eq:tranrate}
\end{equation} where ${\bf 1}_k$ is a $n\times 1$ vector such that the $k$th element is $1$ and the others are $0.$ Note that the first term is for the event that a departure occurs at a queue with size $k$ so $s_k$ decreases by $1/M.$  The second term is for the event that an arrival occurs and it is routed to a queue with size $k-1.$ It has been shown in \cite{VveDobKar_96} that $s^{(M)}_k(\infty)$ converges weakly to $s^*_k,$ where \begin{equation}s^*_k=\lambda^{2^k-1}\label{eq:sd}\end{equation} is the equilibrium point of the following mean-field model:
\begin{eqnarray*}
\dot{s}_k&=&\lambda(s_{k-1}^2-s_k^2)-(s_k-s_{k+1})\quad \forall\  k\geq 1\\
s_0&=&1.
\end{eqnarray*}

\section{Stein's method for the rate of convergence}

Note that the system above is an infinite dimensional system. Analysis of perturbed infinite-dimensional nonlinear systems is a challenging problem. So instead of analyzing the infinite dimensional system, we consider an $n$-dimensional truncated system defined as follows:
\begin{equation}
\dot{\tilde{s}}_k=\left\{
               \begin{array}{ll}
                 \lambda(\tilde{s}_{k-1}^2-\tilde{s}_k^2)-(\tilde{s}_k-\tilde{s}_{k+1}), & n-1\geq k\geq 1; \\
                \lambda(\tilde{s}_{n-1}^2-\tilde{s}_n^2)-(\tilde{s}_n-s^*_{n+1}), &k=n.
               \end{array}
             \right., \label{eq: po2}
\end{equation} where $\tilde{s}_0(t)=1$ for $t\geq 0.$
It is easy to verify that ${\bf s}^*$ defined in (\ref{eq:sd}) remains to be the unique equilibirium point of this truncated system. Furthermore, $\tilde{s}_k(t)\leq 1$ for any $t$ given the initial condition of the system satisfies $1=\tilde{s}_0(0)\geq \tilde{s}_1(0)\geq \tilde{s}_2(0)\geq \cdots\geq \tilde{s}_n(0)$ as shown in Lemma \ref{lem:bound}.
\begin{lemma}
Consider the dynamical system defined in (\ref{eq: po2}). Given that the initial condition satisfies $1=\tilde{s}_0(0)\geq \tilde{s}_1(0)\geq \tilde{s}_2(0)\geq \cdots\geq \tilde{s}_n(0),$ we have $\tilde{s}_k(t)\leq 1$ for any $t\geq 0$ and $0\leq k\leq n.$ \label{lem:bound}
\end{lemma}
\begin{proof}
Note that $\max_{0\leq k\leq n} \tilde{s}_k(0)\leq 1.$ Furthermore, if $$\tilde{s}_j(t)=1=\max_{0\leq k\leq n} \tilde{s}_k(t)\hbox{ and } \sup{0\leq\tau\leq t}\max_{0\leq k\leq n} \tilde{s}_k(\tau)\leq 1$$  then $\dot{\tilde{s}}_j(t)\leq 0.$ Therefore, we conclude that for any $t\geq 0,$ $$\max_{0\leq k\leq n} \tilde{s}_k(t)\leq 1$$ and the lemma holds.
\end{proof}

We next define $x_k=\tilde{s}_k-s_k^*$ for $1\leq k\leq n,$ so
\begin{align}
\dot{x}_k&=f_k({\bf x})\\
&:=\left\{
               \begin{array}{ll}
                 \lambda(\left(x_{k-1}+s_{k-1}^*\right)^2-\left(x_k+s_k^*\right)^2)-(x_k+s_k^*-x_{k+1}-s^*_{k+1}), & 1\leq k\leq n-1\\
                \lambda(\left(x_{n-1}+s_{n-1}^*\right)^2-\left(x_n+s_n^*\right)^2)-(x_n+s_n^*-s^*_{n+1}), & k=n
               \end{array}
             \right.\nonumber\\
&=\left\{
               \begin{array}{ll}
                -\lambda \left(x^2_1+2s_1^*x_1\right)-(x_1-x_2), & k=1\\
             \lambda(\left(x^2_{k-1}+2 s_{k-1}^* x_{k-1} \right)-\left(x^2_k+2s_k^*x_k\right))-(x_k-x_{k+1}), & 2\leq k\leq n-1\\
                \lambda(\left(x_{n-1}^2+2s_{n-1}^*x_{n-1}\right)-\left(x_n^2+2s_n^*x_n\right))-x_n, & k=n
               \end{array}
             \right.\label{ds:truncated}
\end{align}
The unique equilibrium point for the system is ${\bf x}=0.$

\begin{lemma}
Under the dynamical system defined in (\ref{ds:truncated}), $-s_k^*\leq x_k(t)\leq 1-s_k^*$ for $1\leq k \leq n$ and  all $t\geq 0.$\label{lem:xbound}
\end{lemma}
\begin{proof}
According to Lemma \ref{lem:bound}, $0\leq x_k(t)+s_k^*=s_k(t)\leq 1$ for any $t\geq 0,$ so the lemma holds.
\end{proof}

For ease of notation, we always use ${\bf x}$ to denote an infinite-dimensional vector in this section and define $f_k({\bf x})=0$ for $k>n.$ Let $g({\bf x})$ be the solution to the Poisson equation
\begin{equation}
\triangledown g({\bf x})\cdot \dot{\bf x}=\triangledown g({\bf x})\cdot f({\bf x})=\sum_{k=1}^n x_k^2.\label{eq:poisson}
\end{equation}
Then, $$g({\bf x})=-\int_0^\infty \sum_{k=1}^n \left(x_k(t, {\bf x})\right)^2\, dt$$ when the integral is finite \cite{Bar_88,Got_91}, where $x_k(t, {\bf x})$ is the solution of the dynamical system (\ref{ds:truncated}) with ${\bf x}$ as the initial condition. Note that the solution only depends on the first $n$ components of the initial condition $\bf x.$ The integral is finite since the system is exponentially stable as shown in Lemma \ref{lem:expstable} in Section \ref{sec:pert} and $\sum_{k=1}^n |x_k(0)|\leq n$ according to Lemma \ref{lem:xbound}.

Now let $G_M$ denote the generator of the $M$th CTMC. Define $$Q_{{\bf x},{\bf y}}=Q_{{{\bf x}+{\bf s}^*,{\bf y}+{\bf s}^*}}$$ and $$q_{{\bf x},{\bf y}}=\frac{1}{M}Q_{{\bf x},{\bf y}}.$$ Then,
\begin{eqnarray*}
G_Mg({\bf x}) &=& \sum_{{\bf y}: {\bf y}\not={\bf x}}Q_{{\bf x},{\bf y}}({\bf x})\left(g({\bf y})-g({\bf x})\right)\\
&=&M\sum_{{\bf y}: {\bf y}\not={\bf x}}q_{{\bf x},{\bf y}}({\bf x})\left(g({\bf y})-g({\bf x})\right).
\end{eqnarray*}

It has been proved in \cite{Mit_96,VveDobKar_96} that ${\bf x}^{(M)}$ has a stationary distribution.  We use $E_M[\cdot]$ throughout to denote the stationary expectation of the system with $M$ servers. Based on the basic adjoint relation (BAR) \cite{GlyZee_08},
\begin{eqnarray}
E_M\left[G_Mg({\bf x})\right]
=E_M\left[M\sum_{{\bf y}: {\bf y}\not={\bf x}}q_{{\bf x},{\bf y}}({\bf x})\left(g({\bf y})-g({\bf x})\right)\right]=0.\label{eq:ss}
\end{eqnarray}
Then by taking expectation of the Poisson equation (\ref{eq:poisson}) and adding (\ref{eq:ss}) to the equation, we obtain
\begin{eqnarray}
E_M\left[\sum_{k=1}^n x_k^2\right]=E_M\left[\triangledown g({\bf x})\cdot f({\bf x})-M\sum_{{\bf y}: {\bf y}\not={\bf x}}q_{{\bf x},{\bf y}}({\bf x})\left(g({\bf y})-g({\bf x})\right)\right].\nonumber
\end{eqnarray}

From the transition rates of the CTMC under the power-of-two-choices (\ref{eq:tranrate}), we obtain $$\sum_{{\bf y}: {\bf y}\not={\bf x}}q_{{\bf x},{\bf y}}M(y_k-x_k)=\lambda\left(\left(x^2_{k-1}+2 s_{k-1}^* x_{k-1} \right)-\left(x^2_k+2s_k^*x_k\right)\right)-(x_k-x_{k+1}).$$ The definition of $f_k(x)$ (\ref{ds:truncated}) further implies  $$f_k({\bf x})=\left\{
                                         \begin{array}{ll}
                                           \sum_{{\bf y}: {\bf y}\not={\bf x}}q_{{\bf x},{\bf y}}M(y_k-x_k), & \hbox{ if } 1\leq k <n\\
                                           \sum_{{\bf y}: {\bf y}\not={\bf x}}q_{{\bf x},{\bf y}}M(y_n-x_n)-x_{n+1}, & \hbox{ if } k=n.
                                         \end{array}
                                       \right.,$$ and
\begin{eqnarray*}
\triangledown g({\bf x})\cdot f({\bf x})&=&\sum_{k=1}^n  \frac{\partial g}{\partial x_k} \left( \sum_{{\bf y}: {\bf y}\not={\bf x}}q_{{\bf x},{\bf y}}M(y_k-x_k)\right)-\frac{\partial g}{\partial x_n}({\bf x}){x}_{n+1}\\
&=&\sum_{{\bf y}: {\bf y}\not={\bf x}}q_{{\bf x},{\bf y}}M \sum_{k=1}^n  \frac{\partial g}{\partial x_k} \left( y_k-x_k\right)-\frac{\partial g}{\partial x_n}({\bf x}){x}_{n+1}\\
&=&\sum_{{\bf y}: {\bf y}\not={\bf x}}q_{{\bf x},{\bf y}}M \triangledown g({\bf x})\cdot({\bf y}-{\bf x})-\frac{\partial g}{\partial x_n}({\bf x}){x}_{n+1},
\end{eqnarray*} where the last equality holds because $\frac{\partial g}{\partial x_k}=0$ for $k>n.$
Therefore, we have
\begin{align}
&E_M\left[\sum_{k=1}^n x_k^2\right]\\
=&E_M\left[ -\frac{\partial g}{\partial x_n}({\bf x}){x}_{n+1}-M\sum_{{\bf y}: {\bf y}\not={\bf x}}q_{{\bf x},{\bf y}}({\bf x})\left(g({\bf y})-g({\bf x})-\triangledown g({\bf x})\cdot({\bf y}-{\bf x})\right)\right].
\end{align} Since the mean-field model (\ref{ds:truncated}) is a truncated system, we have the additional term $-\frac{\partial g}{\partial x_n}({\bf x}){x}_{n+1}$ compared to the similar equation in \cite{Yin_16}. In the remaining of this paper, we will derive an upper bound on the mean-square error $E_M\left[\sum_{k=1}^\infty x_k^2\right]$ by first deriving an upper bound on  $E_M\left[\sum_{k=1}^n x_k^2\right].$ The result will be proved based on a sequence of convergence results of the dynamical system (\ref{ds:truncated}) to be proved in Section \ref{sec:pert}.

\begin{lemma} Under the power-of-two-choices and given $n\log n=o(M),$ we have
\begin{eqnarray}
E_M\left[\sum_{k=1}^n x_k^2\right]&=& O\left({\lambda^n} n\log n\right)+O\left(\frac{n^3(\log n)^2}{M}\right).\label{eq:bound2}
\end{eqnarray}
\end{lemma}
\begin{proof}
The proof is based on several convergence properties of the dynamical system (\ref{ds:truncated}). The proofs of these convergence properties can be found in Section \ref{sec:pert}.

According to Lemma \ref{lem:expstable} and Corollary \ref{cor:1st-decay} in Section \ref{sec:pert}, both $\sum_{k=1}^n |x_k(t, {\bf x})|$  and $\sum_{k=1}^n |\triangledown x_k(t, {\bf x})|$ decay exponentially as $t$ increases. Furthermore, $|{\bf y}-{\bf x}|=1/M$ for any $\bf x$ and $\bf y$ such that $Q_{\bf x, \bf y}\not=0$ according to (\ref{eq:tranrate}).  Therefore  $$\int_0^\infty 2x_k(t, {\bf x})\triangledown x_k(t, {\bf x})\cdot({\bf y}-{\bf x})\,dt$$ is finite, and by exchanging the order of integration and differentiation, we obtain \begin{eqnarray}
\triangledown g({\bf x})\cdot({\bf y}-{\bf x})=\int_0^\infty \sum_{k=1}^n 2x_k(t, {\bf x})\triangledown x_k(t, {\bf x})\cdot({\bf y}-{\bf x})\,dt.
\end{eqnarray}  Hence, we have
\begin{align}
&-\left(g({\bf y})-g({\bf x})-\triangledown g({\bf x})\cdot({\bf y}-{\bf x})\right)\nonumber\\
=&\int_0^\infty \sum_{k=1}^n \left(\left(x_k(t, {\bf y})\right)^2-\left(x_k(t, {\bf x})\right)^2-2x_k(t, {\bf x})\triangledown x_k(t, {\bf x})\cdot({\bf y}-{\bf x})\right)\,dt.\label{eq:2nd-deri}
\end{align}

We now define \begin{equation*}
e_k(t)=x_k(t, {\bf y})-x_k(t, {\bf x})-\triangledown x_k(t, {\bf x})\cdot({\bf y}-{\bf x}),
\end{equation*} i.e.,
\begin{equation*}
x_k(t, {\bf y})=e_k(t)+x_k(t, {\bf x})+\triangledown x_k(t, {\bf x})\cdot({\bf y}-{\bf x}),
\end{equation*}
so
\begin{eqnarray*}
&&\left(x_k(t, {\bf y})\right)^2-\left(x_k(t, {\bf x})\right)^2-2x_k(t, {\bf x})\triangledown x_k(t, {\bf x})\cdot({\bf y}-{\bf x})\\
&=&\left(e_k(t)+x_k(t, {\bf x})+\triangledown x_k(t, {\bf x})\cdot({\bf y}-{\bf x})\right)^2-\left(x_k(t, {\bf x})\right)^2-2x_k(t, {\bf x})\triangledown x_k(t, {\bf x})\cdot({\bf y}-{\bf x})\\
&=&e^2_k(t)+\left(\triangledown x_k(t, {\bf x})\cdot({\bf y}-{\bf x})\right)^2+2e_k(t)\triangledown x_k(t, {\bf x})\cdot({\bf y}-{\bf x})+2e_k(t)x_k(t, {\bf x})\\
&=&e_k(t)\left(e_k(t)+2\triangledown x_k(t, {\bf x})\cdot({\bf y}-{\bf x})+2x_k(t, {\bf x})\right)+\left(\triangledown x_k(t, {\bf x})\cdot({\bf y}-{\bf x})\right)^2.
\end{eqnarray*}

According to Lemmas \ref{lem:etsmallt} and \ref{lem:etlarget} in Section \ref{sec:pert}, $$|e_k(t)|\leq \sum_{k=1}^n |{e}_k(t)|=O\left(\frac{n\log n}{M^2}\right).$$ Under the power-of-two-choices algorithm, $\sum_{k=1}^{\infty} |x_k-y_k|= \frac{1}{M}$ for any ${\bf x}$ and $\bf y$ such that $Q_{{\bf x},{\bf y}}\not=0.$ Furthermore, $|\triangledown x_k(t, {\bf x})|$ is bounded by $n$ according to Corollary \ref{cor:1st-decay} and
$|x_k(t, {\bf x})|$ is bounded by $|{\bf x}(0)|\leq n$ according to Lemma \ref{lem:exstable}. Therefore, given $n\log n=o(M),$ e.g., $n=\Theta(\log M),$ we can choose a sufficiently large $\tilde{M}$ such that for any $M\geq \tilde{M},$
$$\left|e_k(t)+2\triangledown x_k(t, {\bf x})\cdot({\bf y}-{\bf x})+2\left(x_k(t, {\bf x})\right)\right|\leq 1+2|x_k(t, {\bf x})|\leq 1+2|{\bf x}(t, {\bf x})|\leq 3n,$$ where the last inequality follows from Lemma \ref{lem:exstable}, which implies that
\begin{align}
&\left|g({\bf y})-g({\bf x})-\triangledown g({\bf x})\cdot({\bf y}-{\bf x})\right|\nonumber\\
\leq& 3n{\int}_0^\infty \sum_k|{e}_k(t)|\,dt+\int_0^\infty \sum_k \left(\triangledown x_k(t, {\bf x})\cdot({\bf y}-{\bf x})\right)^2\,dt.\label{eq:boundong}
\end{align}  In Theorem \ref{thm:errorbound},  we will establish that $${\int}_0^\infty \sum_{k=1}^n |{e}_k(t)|\,dt=O\left(\frac{(n\log n)^2}{M^2}\right).$$
Under the power-of-two-choices,  for any $\bf x$ and $\bf y$ such that $Q_{{\bf x}, {\bf y}}\not=0,$ there exists $j$ such that $|x_j-y_j|=1/M$ and $|x_h-x_h|=0$ for $h\not=j.$  Therefore,
\begin{align*}
&\int_0^\infty \sum_{k=1}^n \left(\triangledown x_k(t, {\bf x})\cdot({\bf y}-{\bf x})\right)^2\,dt = \frac{1}{M^2} \int_0^\infty \sum_{k=1}^n \left(\frac{\partial}{\partial x_j}x_k(t, {\bf x})\right)^2\,dt\\
&\leq  \frac{1}{M^2} \int_0^\infty  \left( \sum_{k=1}^n\left|\frac{\partial}{\partial x_j}x_k(t, {\bf x})\right|\right)^2\,dt
\end{align*}
From Corollary \ref{cor:1st-decay}, we have
\begin{eqnarray*}
\int_0^\infty  \left( \sum_{k=1}^n\left|\frac{\partial}{\partial x_j}x_k(t, {\bf x})\right|\right)^2\,dt=O\left(n\log n\right).
\end{eqnarray*} Therefore, we can conclude that
\begin{eqnarray}
\left|g({\bf y})-g({\bf x})-\triangledown g({\bf x})\cdot({\bf y}-{\bf x})\right|=O\left(\frac{n^3(\log n)^2}{M^2}\right).
\label{eq:order}
\end{eqnarray}

Furthermore, from Corollary \ref{cor:1st-decay}, we have
\begin{eqnarray*}
\left|\frac{\partial g}{\partial x_n}({\bf x})\right|&=&\left|-\int_0^\infty \sum_{k=1}^n 2x_k(t, {\bf x}) \frac{\partial}{\partial x_n}x_k(t, {\bf x})\, dt\right|\\
&\leq& 2\left|\int_0^\infty \sum_{k=1}^n \left|\frac{\partial}{\partial x_n}x_k(t, {\bf x})\right|\, dt\right|\\
&=&O(n\log n).
\end{eqnarray*} It has been shown in \cite{VveDobKar_96} that $E_M\left[s_{n}\right]\leq \lambda^n$ for any $n\geq 0,$ so when $n$ is sufficiently large, $$E_M\left[|{x}_{n+1}|\right]\leq E_M\left[|s_{n+1}|\right]+s^*_{n+1}\leq \lambda^{n+1}+\lambda^{2^{n+1}-1}\leq \lambda^n.$$ Therefore, we conclude
\begin{eqnarray*}
E_M\left[\sum_{k=1}^n \left(x_k\right)^2\right]&=& O\left({\lambda^n} n\log n\right)+O\left(\frac{n^3(\log n)^2}{M}\right).
\end{eqnarray*}
\end{proof}

From the lemma above, we can conclude the following theorem.
\begin{theorem}
Under the power-of-two-choices with $\lambda<1,$
$$E_M\left[\sum_{k=1}^\infty \left|s^{(M)}_k-s_k^*\right|^2\right]=E_M\left[\sum_{k=1}^\infty x_k^2\right]=O\left(\frac{(\log M)^3(\log\log M)^2}{M}\right).$$
\end{theorem}
\begin{proof}
By choosing $n=\frac{3\log M}{\log \frac{1}{\lambda}}$ in equation (\ref{eq:bound2}), we obtain
 $$E_M\left[\sum_{k=1}^{n} x_k^2\right]=O\left(\frac{(\log M)^3(\log\log M)^2}{M}\right).$$
Since  $E_M\left[s_k\right]\leq \lambda^k$ and $0\leq s_k\leq 1$ for all $k,$
$$E_M\left[\sum_{k=n+1}^\infty x_k^2\right]\leq E_M\left[\sum_{k=n+1}^\infty s_k^2+\left(s^*_k\right)^2\right]\leq E_M\left[\sum_{k=n+1}^\infty s_k+\left(s^*_k\right)^2\right]\leq \frac{\lambda^{n}}{1-\lambda},$$ where the last inequality holds when $n$ is sufficiently large. The theorem holds.
\end{proof}

\section{Convergence Properties of the Dynamical System (\ref{ds:truncated})}
\label{sec:pert}

In this section, we analyze the dynamical system defined by (\ref{ds:truncated}), and present the convergence properties used in the previous section. Since the system is an $n$-dimensional truncated system of the original infinite-dimensional dynamical system, in this section, all vectors are $n$-dimensional instead of infinite-dimensional.

\begin{lemma}
Under the dynamical system defined in (\ref{ds:truncated}), we have for any $t\geq 0,$ $$|{\bf x}(t)|\leq |{\bf x}(0)|.$$\label{lem:exstable}
\end{lemma}
\begin{proof}
The proof follows the proof of Theorem 3.6 in \cite{Mit_96}. Define the Lyapunov function to be $$V(t)=\sum_{k=1}^n |x_k(t)|.$$ Note that $\frac{d|x_k(t)|}{dt}$ is well-defined for $x_k(t)\not=0.$ When $x_k(t)=0,$ we consider the upper right-hand derivative as in \cite{Mit_96}, i.e.,
\begin{eqnarray*}
\frac{d|x_k(t)|}{dt}=\lim_{\tau\rightarrow t^+}\frac{|x_k(\tau)|-|x_k(t)|}{\tau-t}=\left\{
                                                                                     \begin{array}{ll}
                                                                                       \dot{x}_k(t), & \hbox{ if } \dot{x}_k(t)\geq 0\\
                                                                                       -\dot{x}_k(t), & \hbox{ if }\dot{x}_k(t)< 0
                                                                                     \end{array}
                                                                                   \right..
\end{eqnarray*}

Consider $k$ such that $n>k>1.$ If $x_k(t)>0,$ or $x_k(t)=0$ and $\dot{x}_k\geq 0,$ then
\begin{eqnarray*}
\frac{d|x_k(t)|}{dt}=\frac{dx_k(t)}{dt}&=&\lambda\left(\left(x^2_{k-1}+2 s_{k-1}^* x_{k-1} \right)-\left(x^2_k+2s_k^*|x_k|\right)\right)-(|x_k|-x_{k+1})\\
&=&\lambda\left(x^2_{k-1}+2 s_{k-1}^* x_{k-1} \right)-\lambda \left(x^2_k+2s_k^*|x_k|\right)-|x_k|+x_{k+1}.
\end{eqnarray*}
Define $\dot{V}(t)=\sum_{k=1}^n W_k(t)$ such that $W_k(t)$ includes all the terms involving $x_k(t).$ Then when $x_k(t)>0,$ or $x_k(t)=0$ and $\dot{x}_k\geq 0,$ we have
\begin{eqnarray*}
W_k(t)&\leq &\lambda\left(x^2_{k}+2 s_{k}^* |x_{k}| \right)-\lambda \left(x^2_k+2s_k^*|x_k|\right)-|x_k|+|x_{k}|=0,
\end{eqnarray*} where the first term comes from $\frac{d|x_{k+1}(t)|}{dt}$ and the last term comes from $\frac{d|x_{k-1}(t)|}{dt}.$

If $x_k(t)<0,$  or $x_k(t)=0$ and $\dot{x}_k< 0,$ then
\begin{eqnarray*}
\frac{d|x_k(t)|}{dt}=-\frac{dx_k(t)}{dt}&= &-\lambda\left(x^2_{k-1}+2 s_{k-1}^* x_{k-1} \right)+\lambda \left(x^2_k-2s_k^*|x_k|\right)-|x_k|-x_{k+1},
\end{eqnarray*} so
\begin{eqnarray*}
W_k(t)&\leq &\lambda\left(-x^2_{k}+2 s_{k}^* |x_{k}| \right)+\lambda \left(x^2_k-2s_k^*|x_k|\right)-|x_k|+|x_{k}|=0,
\end{eqnarray*} where the last term comes from $\frac{d|x_{k-1}(t)|}{dt},$ and the first term comes from $\frac{d|x_{k+1}(t)|}{dt}$ and the fact that $x_k(t)\geq -s_k^*$ when $x_k(t)< 0,$ so $\left(-x^2_{k}+2 s_{k}^* |x_{k}|\right)>0.$

Similarly, it can be shown that $W_k(t)\leq 0$ for $k=1$ and $k=n.$ From the discussion above, we have $$\dot{V}(t)=\sum_{k=1}^n W_k \leq 0$$ for all $t\geq 0$ and the lemma holds.
\end{proof}

Define $$k_{\lambda}=\left\lceil \frac{\log\left(\frac{\log\left(\frac{1-\sqrt{\lambda}}{2\sqrt{\lambda}}\right)}{\log\lambda}+1\right)}{\log 2}\right\rceil,$$ where implies that $$\lambda\left(2s^*_{k_{\lambda}}+1\right)=\lambda\left(2\lambda^{2^{k_{\lambda}}-1}+1\right)\leq\sqrt{\lambda}.$$ Furthermore, define a sequence $\{w_k\}$ such that
\begin{eqnarray*}
w_0&=&0\\
w_1&=&1\\
w_k&=&\left(1+\frac{1}{k_{\lambda}}\sum_{j=1}^k\frac{1}{\left(\max\{1,4\lambda\}\right)^{j-1} }\right), \quad 2\leq k \leq k_{\lambda}\\
w_{k}&=&\left(1+\frac{k-k_{\lambda}}{n}\right)w_{k_{\lambda}},\quad k_{\lambda}<k\leq n,
\end{eqnarray*} and
$$\delta=\frac{1-\sqrt{\lambda}}{4n}.$$

\begin{lemma}
Under the dynamical system defined in (\ref{ds:truncated}), we have for any $t\geq 0,$ $$\sum_{k=1}^n w_k |x_k(t)|\leq \left(\sum_{k=1}^n w_k |x_k(0)|\right)e^{-\delta t}.$$ \label{lem:expstable}
\end{lemma}
\begin{proof} The proof follows the idea  of Theorem 3.6 in \cite{Mit_96}. Define $V(t)=\sum_{k=1}^n w_k |x_k(t)|$   and  $\dot{V}(t)=\sum_{k=1}^n W_k(t)$ such that $W_k(t)$ includes all the terms involving $x_k(t).$ The lemma is proved by showing that \begin{equation} W_k(t)\leq -\delta w_k |x_k(t)|.\label{eq:wex}\end{equation} We consider the case where $x_k(t)>0,$ or $x_k(t)=0$ and $\dot{x}_k\geq 0.$ According to the proof of Lemma \ref{lem:exstable}, we have
\begin{eqnarray*}
W_k(t)&\leq &w_{k+1}\lambda\left(x^2_{k}+2 s_{k}^* |x_{k}| \right)-w_k\lambda \left(x^2_k+2s_k^*|x_k|\right)-w_k|x_k|+w_{k-1}|x_{k}|.
\end{eqnarray*}
So (\ref{eq:wex}) holds if $$w_{k+1}\lambda\left(x^2_{k}+2 s_{k}^* |x_{k}| \right)-w_k\lambda \left(x^2_k+2s_k^*|x_k|\right)-w_k|x_k|+w_{k-1}|x_{k}|\leq -\delta w_k |x_k|,$$ in other words, if
\begin{equation}
w_{k+1} - w_k \leq \frac{(1-\delta) w_k -w_{k-1}}{\lambda(|x_{k}|+2 s_{k}^*)}.\label{eq:ww}
\end{equation}
We now prove (\ref{eq:ww}) by considering the following three cases.

When $1\leq k\leq k_{\lambda}-1,$ we have
\begin{eqnarray*}
w_{k+1}-w_k&=&\frac{1}{k_{\lambda}\left(\max\{1,4\lambda\}\right)^{k}}\\
\frac{(1-\delta) w_k -w_{k-1}}{\lambda(|x_{k}|+2 s_{k}^*)}&\geq&\frac{ w_k -w_{k-1}-\delta w_k}{3\lambda}=\frac{\frac{1}{k_{\lambda}\left(\max\{1,4\lambda\}\right)^{k-1}}-\delta w_k}{3\lambda}.
\end{eqnarray*} So inequality (\ref{eq:ww}) holds if
\begin{eqnarray*}
3\lambda&\leq& \max\{1,4\lambda\}-\delta w_k k_{\lambda}\left(\max\{1,4\lambda\}\right)^{k},
\end{eqnarray*} which can be established by proving
\begin{eqnarray*}
\delta w_k k_{\lambda} \left(\max\{1,4\lambda\}\right)^{k}&\leq& \lambda.
\end{eqnarray*} Since $\delta=\frac{1-\sqrt{\lambda}}{4n}$ and $w_k\leq 2$ for $k\leq k_{\lambda}-1,$ the inequality above holds when $n$ is sufficiently large.

When $n\geq k\geq k_{\lambda}+1,$ according to the definition of $k_{\lambda},$ $\lambda(|x_k|+2s_k^*)\leq \sqrt{\lambda}.$ Therefore,  we have
\begin{eqnarray*}
w_{k+1}-w_k&=&\frac{w_{k_{\lambda}}}{n}\\
\frac{(1-\delta) w_k -w_{k-1}}{\lambda(|x_{k}|+2 s_{k}^*)}&\geq&\frac{w_k -w_{k-1}-\delta w_k}{\sqrt{\lambda}}=\frac{\frac{w_{k_{\lambda}}}{n}-\delta w_k}{\sqrt{\lambda}}.
\end{eqnarray*} So inequality (\ref{eq:ww}) holds if
\begin{eqnarray*}
\sqrt{\lambda}w_{k_{\lambda}}&\leq& w_{k_{\lambda}}-\delta w_k n,
\end{eqnarray*} in other words, if
\begin{eqnarray*}
w_k&\leq& \frac{(1-\sqrt{\lambda})w_{k_{\lambda}}}{\delta n}=4w_{k_{\lambda}}.
\end{eqnarray*} Since $w_k\leq 4$ and $w_{k_{\lambda}}>1,$ the inequality above holds.

When $k=k_{\lambda},$ according to the definition of $k_{\lambda},$ $\lambda(|x_k|+2s_k^*)\leq \sqrt{\lambda}.$ Therefore,  we have
\begin{eqnarray*}
w_{k+1}-w_k&=&\frac{w_{k_{\lambda}}}{n}\\
\frac{(1-\delta) w_k -w_{k-1}}{\lambda(|x_{k}|+2 s_{k}^*)}&\geq&\frac{w_k -w_{k-1}-\delta w_k}{\sqrt{\lambda}}=\frac{\frac{1}{k_{\lambda}\left(\max\{1,4\lambda\}\right)^{k_{\lambda}-1}}-\delta w_k}{\sqrt{\lambda}}.
\end{eqnarray*} So inequality (\ref{eq:ww}) holds if
\begin{eqnarray*}
\sqrt{\lambda}w_{k_{\lambda}}&\leq& \frac{n}{k_{\lambda}\left(\max\{1,4\lambda\}\right)^{k_{\lambda}-1}}-\delta w_k n,
\end{eqnarray*} in other words, if
\begin{eqnarray*}
w_k&\leq& \frac{4}{1-\sqrt{\lambda}}\left(\frac{n}{k_{\lambda}\left(\max\{1,4\lambda\}\right)^{k_{\lambda}-1}}-\sqrt{\lambda}w_{k_{\lambda}}\right).
\end{eqnarray*} Since $w_k\leq 4$ the inequality above holds when $n$ is sufficiently large.

From the discussion above, we conclude that
$$\dot{V}(t)\leq -\sum_{k=1}^n \delta w_k |x_k(t)|=-\delta V(t),$$ so the lemma holds.
\end{proof}

For the ease of notation, in the following analysis, we use ${\bf x}(t,0)={\bf x}(t, {\bf x})$ and ${\bf x}(t,\epsilon)={\bf x}(t, {\bf y}),$ where $\epsilon=1/M,$ which is also to be consistent with notation used in \cite{Kha_01} such that ${\bf x}(t,\epsilon)$ is the perturbed version of ${\bf x}(t,0).$   We consider the finite Taylor series for ${\bf x}(t, \epsilon):={\bf x}(t, {\bf y})$ in terms of $\epsilon:$
\begin{equation}
{\bf x}(t,\epsilon)={\bf x}^{(0)}(t) +\epsilon {\bf x}^{(1)}(t)+ {\bf e}(t),\label{eq:taylor}
\end{equation} and
\begin{equation}
{\bf x}(0,\epsilon)={\bf x}+\epsilon {\bf z},
\end{equation} where $${\bf x}^{(0)}(t)={\bf x}(t,0) \hbox{ and } {\bf x}^{(1)}(t)=\left.\frac{d {\bf x}}{d \epsilon}(t,\epsilon)\right|_{\epsilon=0},$$ and ${\bf z}=M\left({\bf y}-{\bf x}\right)$ (again, all three vectors are $n$-dimensional).
Substituting (\ref{eq:taylor}) into the dynamical system equation, we get
\begin{eqnarray}
\dot{\bf x}(t,\epsilon)&=&\dot{\bf x}^{(0)}(t) +\epsilon \dot{\bf x}^{(1)}(t)+ \dot{\bf e}(t)=f({\bf x}(t,\epsilon))\\
&=& h^{(0)}({\bf x}^{(0)}(t))+h^{(1)}({\bf x}^{(\leq 1)}(t))\epsilon+R_{\bf e}(t,\epsilon),
\end{eqnarray} where ${\bf x}^{(\leq 1)}=\left({\bf x}^{(0)}, {\bf x}^{(1)}\right).$ The zero-order term $h^{(0)}$ is given by
$$\dot{x}^{(0)}(t)=h^{(0)}\left({\bf x}^{(0)}(t)\right)=f\left({\bf x}^{(0)}(t)\right)\quad \hbox{with}\quad{\bf x}^{(0)}(0)={\bf x},$$ which is the nominal system without the perturbation on the initial condition. The first-order term is given by
\begin{eqnarray*}
h^{(1)}\left({\bf x}^{({\leq 1})}(t)\right)&=&\left.\frac{d}{d \epsilon}f({\bf x}(t,\epsilon))\right|_{\epsilon=0}\\
&=&\left. \frac{\partial f}{\partial x}({\bf x}(t,\epsilon))\frac{d {\bf x}}{d \epsilon}(t,\epsilon)\right|_{\epsilon=0}\\
&=&\frac{\partial f}{\partial x}({\bf x}^{(0)}(t)){\bf x}^{(1)}(t).
\end{eqnarray*}
Therefore, we have
\begin{equation}\dot{\bf x}^{(1)}(t)=\frac{\partial f}{\partial x}({\bf x}^{(0)}(t)){\bf x}^{(1)}(t)\quad \hbox{with}\quad{\bf x}^{(1)}(0)={\bf z}, \label{ds:1st-def}
\end{equation}
which implies that
\begin{align}
&\dot{x}_k^{(1)}=g_k({\bf x}^{(1)}_1)=\nonumber\\
&\left\{
                                  \begin{array}{ll}
 - 2\lambda \left(x^{(0)}_1+s^*_k\right)x^{(1)}_1-x^{(1)}_1+x^{(1)}_2, & k=1\\
                                    2\lambda \left(x^{(0)}_{k-1}+s^*_{k-1}\right)x^{(1)}_{k-1}- 2\lambda \left(x^{(0)}_{k}+s^*_k\right)x^{(1)}_{k}-x^{(1)}_{k}+x^{(1)}_{k+1}, & 2\leq k\leq n-1 \\
                                    2\lambda \left(x^{(0)}_{n-1}+s_{n-1}^*\right)x^{(1)}_{n-1}- 2\lambda \left(x^{(0)}_{n}+s_n^*\right)x^{(1)}_{n}-x^{(1)}_{n}, & k=n
                                  \end{array}
                                \right..\label{1st-sys}
\end{align}

\begin{lemma} Under the dynamical system defined by (\ref{1st-sys}),
\begin{eqnarray*}
\left|{\bf x}^{(1)}(t)\right|\leq \left|{\bf x}^{(1)}(0)\right|.
\end{eqnarray*}\label{lem:x1bound}
\end{lemma}
\begin{proof}
First recall that $x^{(0)}_{k}(t)+s^*_k\geq 0$ for any $t\geq 0$ and $k$ according to Lemma \ref{lem:xbound}. Define $$V(t)=\sum_{k=1}^n  \left|x^{(1)}(t)\right|.$$ Following the proof of Lemma \ref{lem:exstable}, we obtain that
\begin{align}
&\frac{d|{x}_k^{(1)}(t)|}{dt}\leq \nonumber\\
&\left\{
                                  \begin{array}{ll}
 - 2\lambda \left(x^{(0)}_1+s^*_k\right)\left|x^{(1)}_1\right|-\left|x^{(1)}_1\right|+\left|x^{(1)}_2\right|, & k=1\\
                                    2\lambda \left(x^{(0)}_{k-1}+s^*_{k-1}\right)\left|x^{(1)}_{k-1}\right|- 2\lambda \left(x^{(0)}_{k}+s^*_k\right)\left|x^{(1)}_{k}\right|-\left|x^{(1)}_{k}\right|+\left|x^{(1)}_{k+1}\right|, & 2\leq k\leq n-1 \\
                                    2\lambda \left(x^{(0)}_{n-1}+s_{n-1}^*\right)\left|x^{(1)}_{n-1}\right|- 2\lambda \left(x^{(0)}_{n}+s_n^*\right)\left|x^{(1)}_{n}\right|-\left|x^{(1)}_{n}\right|, & k=n
                                  \end{array}
                                \right..\label{1st-sys-b}
\end{align} Therefore
$$\dot{V}(t)\leq  -2\lambda \left(x^{(0)}_{n}+s_n^*\right)\left|x^{(1)}_{n}\right|-\left|x^{(1)}_1\right|\leq 0,$$ and the lemma holds.
\end{proof}

Define $\tilde{k}=\left\lceil\frac{\log\left(\frac{\log 8}{\log\frac{1}{\lambda}}+1\right)}{\log 2}\right\rceil,$ a sequence of $\tilde{w}_k$ such that
\begin{eqnarray*}
\tilde{w}_1&=&1\\
\tilde{w}_k&=&\left(1+\frac{1}{\tilde{k}}\sum_{j=1}^k\frac{1}{\left(\max\{1,5\lambda\}\right)^{j-1} }\right), \quad 2\leq k \leq \tilde{k}\\
\tilde{w}_{k}&=&\left(1+\frac{k-\tilde{k}}{n}\right)\tilde{w}_{\tilde{k}},\quad \tilde{k}<k\leq n.
\end{eqnarray*}
Furthermore, we define $$\tilde{\delta}=\frac{2-\lambda}{8n}$$ and
$$\tilde{t}= \frac{1}{\delta}\log\left(32n\right)\geq \frac{1}{\delta}\log\left(32|{\bf x}(0)|\right).$$

\begin{lemma}
There exists $\tilde{n}$ independent of $M$ such that for any $n\geq \tilde{n}$ and $t\geq \tilde{t},$ we have
\begin{eqnarray*}
\sum_{k=1}^n \tilde{w}_k\left|x^{(1)}_k(t)\right|\leq \sum_{k=1}^n \tilde{w}_k \left|x^{(1)}_k(\tilde{t})\right|e^{-\tilde{\delta} (t-\tilde{t})}.
\end{eqnarray*} \label{lem:x1}
\end{lemma}
\begin{proof}
Since $\tilde{k}=\left\lceil\frac{\log\left(\frac{\log 8}{\log\frac{1}{\lambda}}+1\right)}{\log 2}\right\rceil,$ so $s^*_k\leq s^*_{\tilde{k}}=\lambda^{2^{\tilde{k}}-1}\leq \frac{1}{8}$ for any $k\geq \tilde{k}.$ Now according to Lemma \ref{lem:expstable} and the fact that $1\leq w_k\leq 4,$ we have
$$\sum_{k=1}^n |x^{(0)}_k(t)| \leq \sum_{k=1}^n w_k |x^{(0)}_k(t)|\leq \left(\sum_{k=1}^n w_k |x_k(0)|\right)e^{-{\delta} t}\leq 4 \left(\sum_{k=1}^n |x_k(0)|\right)e^{-\delta t}.$$ Therefore, when $t\geq \tilde{t}\geq \frac{1}{\delta}\log\left(32|{\bf x}(0)|\right),$ $$|{\bf x}^{(0)}(t)|\leq \frac{1}{8}.$$

We further define the following Lyapunov function
\begin{eqnarray*}
{V}({\bf x}^{(1)})=\sum_{k=1}^n \tilde{w}_k \left|x^{(1)}_k\right|.
\end{eqnarray*}
Following the proof of Lemma \ref{lem:exstable} or the proof of Theorem 3.6 in \cite{Mit_96}, we obtain that
\begin{eqnarray*}
\dot{V}({\bf x}^{(1)}) &\leq&\sum_{k=1}^n -\left(2\lambda \tilde{w}_k\left(x^{(0)}_{k}+s^*_k\right)+\tilde{w}_k-2\lambda \tilde{w}_{k+1}\left(x^{(0)}_{k}+s^*_{k}\right)-\tilde{w}_{k-1}\right) \left|x^{(1)}_k\right|.
\end{eqnarray*}
So the lemma holds by proving
\begin{eqnarray*}
-\left(2\lambda \tilde{w}_k\left(x^{(0)}_{k}+s^*_k\right)+\tilde{w}_k-2\lambda \tilde{w}_{k+1}\left(x^{(0)}_{k}+s^*_{k}\right)-\tilde{w}_{k-1}\right)\leq -\tilde{\delta} \tilde{w}_k,
\end{eqnarray*} i.e., by proving
\begin{eqnarray}
\tilde{w}_{k+1}-\tilde{w}_k\leq \frac{\tilde{w}_k-\tilde{w}_{k-1}-\tilde{\delta} \tilde{w}_k}{2\lambda \left(x^{(0)}_{k}+s^*_k\right)}.
\label{eq:www}
\end{eqnarray}
We now prove (\ref{eq:www}) by considering the following three cases.

When $1\leq k\leq \tilde{k}-1,$ we have
\begin{eqnarray*}
\tilde{w}_{k+1}-\tilde{w}_k&=&\frac{1}{\tilde{k}\left(\max\{1,5\lambda\}\right)^{k}}\\
\frac{\tilde{w}_k-\tilde{w}_{k-1}-\tilde{\delta} \tilde{w}_k}{2\lambda \left(x^{(0)}_{k}+s^*_k\right)}&\geq&\frac{\frac{1}{\tilde{k}\left(\max\{1,5\lambda\}\right)^{k-1}}-\tilde{\delta} \tilde{w}_k}{4\lambda}.
\end{eqnarray*} So inequality (\ref{eq:ww}) holds if
\begin{eqnarray*}
4\lambda&\leq& \max\{1,5\lambda\}-\tilde{\delta} \tilde{w}_k \tilde{k}\left(\max\{1,5\lambda\}\right)^{k},
\end{eqnarray*} which can be established by proving
\begin{eqnarray*}
\tilde{\delta} \tilde{w}_k \tilde{k}\left(\max\{1,5\lambda\}\right)^{k}&\leq& \lambda.
\end{eqnarray*} Since $\tilde{\delta}=\frac{2-\lambda}{8n}$ and $\tilde{w}_k\leq 2$ for $k\leq \tilde{k}-1,$ the inequality above holds when $n$ is sufficiently large.

When $n\geq k\geq \tilde{k}+1,$ according to the definition of $\tilde{k},$ $s_k^*\leq \frac{1}{8}.$  Furthermore, given $t\geq \tilde{t},$ $|x^{(0)}_{k}|\leq \frac{1}{8}$ for any $k.$ Therefore,  we have
\begin{eqnarray*}
\tilde{w}_{k+1}-\tilde{w}_k&=&\frac{\tilde{w}_{\tilde{k}}}{n}\\
\frac{(1-\delta) \tilde{w}_k -\tilde{w}_{k-1}}{2\lambda(|x_{k}|+s_{k}^*)}&\geq&\frac{\tilde{w}_k -\tilde{w}_{k-1}-\delta \tilde{w}_k}{\frac{\lambda}{2}}=\frac{\frac{\tilde{w}_{\tilde{k}}}{n}-\delta \tilde{w}_k}{\frac{\lambda}{2}}.
\end{eqnarray*} So inequality (\ref{eq:www}) holds if
\begin{eqnarray*}
\frac{1}{2}{\lambda}\tilde{w}_{\tilde{k}}&\leq& \tilde{w}_{\tilde{k}}-\tilde{\delta} \tilde{w}_k n,
\end{eqnarray*} in other words, if
\begin{eqnarray*}
\tilde{w}_k&\leq& \frac{\left(1-\frac{\lambda}{2}\right)\tilde{w}_{\tilde{k}}}{\tilde{\delta} n}=\frac{\left(8-{4\lambda}\right)\tilde{w}_{\tilde{k}}}{2-\lambda}\leq 4\tilde{w}_{\tilde{k}}.
\end{eqnarray*} Since $\tilde{w}_k\leq 4$ and $\tilde{w}_{\tilde{k}}>1$ by the definitions, (\ref{eq:www}) holds.

When $k=\tilde{k},$ according to the definition of $\tilde{k}$ and $\tilde{t},$ we have $2\lambda(|x_k|+s_k^*)\leq \frac{\lambda}{2}$ for $t\geq \tilde{t}.$ Therefore,  we have
\begin{eqnarray*}
\tilde{w}_{\tilde{k}+1}-\tilde{w}_{\tilde{k}}&=&\frac{\tilde{w}_{\tilde{k}}}{n}\\
\frac{(1-\delta) \tilde{w}_{\tilde{k}} -\tilde{w}_{\tilde{k}-1}}{2\lambda(|x_{\tilde{k}}|+ s_{\tilde{k}}^*)}&\geq&\frac{\tilde{w}_{\tilde{k}} -\tilde{w}_{\tilde{k}-1}-\tilde{\delta} \tilde{w}_{\tilde{k}}}{\frac{\lambda}{2}}=\frac{\frac{1}{\tilde{k}\left(\max\{1,5\lambda\}\right)^{\tilde{k}-1}}-\tilde{\delta} \tilde{w}_{\tilde{k}}}{\frac{\lambda}{2}}.
\end{eqnarray*} So inequality (\ref{eq:ww}) holds if
\begin{eqnarray*}
\frac{\lambda}{2}\tilde{w}_{\tilde{k}}&\leq& \frac{n}{\tilde{k}\left(\max\{1,5\lambda\}\right)^{\tilde{k}-1}}-\tilde{\delta} \tilde{w}_k n,
\end{eqnarray*} in other words, if
\begin{eqnarray*}
\left(\frac{\lambda}{2}+\frac{2-\lambda}{8}\right)\tilde{w}_{\tilde{k}}&\leq& \frac{n}{\tilde{k}\left(\max\{1,5\lambda\}\right)^{\tilde{k}-1}}.
\end{eqnarray*} Since $\tilde{w}_{\tilde{k}}\leq 2$ by the definition,  the inequality above holds when $n$ is sufficiently large.

From the analysis above, we conclude when $t\geq \tilde{t},$
\begin{eqnarray*}
\dot{V}(t)\leq  -\tilde{\delta} V(t)
\end{eqnarray*} and
\begin{eqnarray*}
V(t)\leq V(\tilde{t})e^{-\tilde{\delta} (t-\tilde{t})}.
\end{eqnarray*}
\end{proof}

\begin{cor}
\begin{eqnarray*}
\left|{\bf x}^{(1)}(t)\right|\leq \left\{
                                    \begin{array}{ll}
                                      1, & 0\leq t\leq \tilde{t} \\
                                      \min\left\{1, 4 e^{-\tilde{\delta} (t-\tilde{t})}\right\}, & t\geq \tilde{t}
                                    \end{array}
                                  \right..
\end{eqnarray*}\label{cor:x1bound}
\end{cor}
\begin{proof}
Note that $\left|{\bf x}^{(1)}(0)\right|=|{\bf z}| \in\{1, 0\}$ under the power-of-two-choices. The case for $t\leq \tilde{t}$ holds according to Lemma \ref{lem:x1bound}. Since $1\leq \tilde{w}_k\leq 4$ according to its definition, we can further conclude the case when $t\geq \tilde{t}.$
\end{proof}

\begin{cor}
\begin{eqnarray*}
\sum_{k=1}^n \left|\frac{\partial}{y_j} x_k(t, {\bf y})\right|\leq \left\{
                                    \begin{array}{ll}
                                      1, & 0\leq t\leq \tilde{t} \\
                                      \min\left\{1, 4 e^{-\tilde{\delta} (t-\tilde{t})}\right\}, & t\geq \tilde{t}
                                    \end{array}
                                  \right..
\end{eqnarray*}\label{cor:1st-decay}
\end{cor}
\begin{proof}
Notice that $$\frac{\partial}{y_j} x_k(t, {\bf y})=\lim_{\epsilon\rightarrow 0} \frac{x_k(t, {\bf y}+\epsilon {\bf 1}_{j})-x_k(t, {\bf y})}{\epsilon},$$ where   ${\bf 1}_{j}$ is an $n\times 1$ vector such that  ${\bf 1}^{(j)}_j=1$ and  ${\bf 1}^{(j)}_k=0$ for $k\not=j.$ Therefore, $$\frac{\partial}{z_j} x_k(t, {\bf y})={x}_k^{(1)}(t)$$ with ${\bf x}^{(1)}(0)={\bf 1}_{j}$ and ${\bf x}^{(0)}(t)={\bf x}(t, {\bf y}).$ The corollary follows from the corollary above.
\end{proof}

We next study ${\bf e}(t)={\bf x}(t,\epsilon)-{\bf x}^{(0)}(t)-\epsilon{\bf x}^{(1)}(t).$ According to its definition and \eqref{ds:1st-def}, we have
\begin{eqnarray*}
\dot{\bf e}(t)&=&f\left({\bf x}(t,\epsilon)\right)-f\left({\bf x}^{(0)}(t)\right)-\epsilon\frac{\partial f}{\partial x}({\bf x}^{(0)}(t)){\bf x}^{(1)}(t)\\
{\bf e}(0)&=&0.
\end{eqnarray*}

\begin{lemma}
Assume $n\log n=o(M).$ For any $0\leq t\leq \tilde{t},$ $$|{\bf e}(t)|\leq \frac{2\tilde{t}}{M^2}=O\left(\frac{n\log n}{M^2}\right).$$ \label{lem:etsmallt}
\end{lemma}
\begin{proof}
We first have for $1<k<n,$
\begin{eqnarray*}
&&\dot{e}_k(t)\\
&=&f_k\left( {\bf x}^{(0)}(t)+\epsilon{\bf x}^{(1)}(t)+{\bf e}(t)\right)-f_k\left({\bf x}^{(0)}(t)\right)-\epsilon\sum_{j=1}^n \frac{\partial f_k}{\partial x_j}({\bf x}^{(0)}(t)){x}^{(1)}_j(t)\\
&=&\lambda\left(\left( {x}_{k-1}^{(0)}(t)+\epsilon{x}_{k-1}^{(1)}(t)+{e}_{k-1}(t)\right)^2+2 s_{k-1}^*\left( {x}_{k-1}^{(0)}(t)+\epsilon{x}_{k-1}^{(1)}(t)+{e}_{k-1}(t)\right) \right)\\
&&-\lambda\left(\left( {x}_{k}^{(0)}(t)+\epsilon{x}_{k}^{(1)}(t)+{e}_{k}(t)\right)^2+2 s_{k}^*\left( {x}_{k}^{(0)}(t)+\epsilon{x}_{k}^{(1)}(t)+{e}_{k}(t)\right) \right)\\
&&-\left({x}_{k}^{(0)}(t)+\epsilon{x}_{k}^{(1)}(t)+{e}_{k}(t)- {x}_{k+1}^{(0)}(t)-\epsilon{x}_{k+1}^{(1)}(t)-{e}_{k+1}(t)\right)\\
&&-\lambda\left(\left( {x}_{k-1}^{(0)}(t)\right)^2+2 s_{k-1}^* {x}_{k-1}^{(0)}(t) \right)+\lambda\left(\left( {x}_{k}^{(0)}(t)\right)^2+2 s_{k}^* {x}_{k}^{(0)}(t)\right)\\
&&+\left({x}_{k}^{(0)}(t)- {x}_{k+1}^{(0)}(t)\right) \\
&&-2\epsilon\lambda \left(x^{(0)}_{k-1}+s^*_{k-1}\right)x^{(1)}_{k-1}+ 2\epsilon\lambda \left(x^{(0)}_{k}+s^*_k\right)x^{(1)}_{k}+\epsilon x^{(1)}_{k}-\epsilon x^{(1)}_{k+1}\\
&=&\lambda\left(e_{k-1}^2+2\left(x_{k-1}^{(0)}+s_{k-1}^*+\epsilon x^{(1)}_{k-1}\right)e_{k-1}-e_{k}^2-2\left(x_{k}^{(0)}+s_{k}^*+\epsilon x^{(1)}_{k}\right)e_{k}\right)\\
&&-(e_{k}-e_{k+1})+\lambda \epsilon^2 \left(\left(x^{(1)}_{k-1}\right)^2 -\left(x^{(1)}_{k}\right)^2\right)\\
&=&2\lambda\left(x_{k-1}^{(0)}+s_{k-1}^*\right)e_{k-1}-2\lambda\left(x_{k}^{(0)}+s_{k}^*\right)e_{k}-(e_k-e_{k+1})\\
&&+\lambda\left(e_{k-1}^2+2\epsilon x^{(1)}_{k-1}e_{k-1}-e_{k}^2-2\epsilon x^{(1)}_{k}e_{k}\right)+\lambda \epsilon^2 \left(\left(x^{(1)}_{k-1}\right)^2 -\left(x^{(1)}_{k}\right)^2\right)\\
&=&g_k\left({\bf e}\right)+\lambda\left(e_{k-1}^2+2\epsilon x^{(1)}_{k-1}e_{k-1}-e_{k}^2-2\epsilon x^{(1)}_{k}e_{k}\right)+\lambda \epsilon^2 \left(\left(x^{(1)}_{k-1}\right)^2 -\left(x^{(1)}_{k}\right)^2\right),
\end{eqnarray*} where the last equality holds according to the definition of $g_k(\cdot)$ in (\ref{1st-sys}). The same equation holds for $k=1$ and $k=n.$

Define $V(t)=|{\bf e}(t)|.$ Now following the proof of Lemma \ref{lem:exstable}, we can obtain
\begin{eqnarray*}
\dot{V}(t)\leq \sum_{k=1}^n 2\lambda\left(e_{k}^2+2\epsilon \left|x^{(1)}_{k}\right||e_{k}|\right)+2\lambda \epsilon^2\left(x^{(1)}_{k}\right)^2.
\end{eqnarray*} Assume $|{\bf e}(t)|\leq \frac{2\tilde{t}}{M^2}$ for $t\leq \tilde{t},$ we have that for sufficiently large $M,$
\begin{eqnarray*}
\dot{V}(t)\leq \frac{10\lambda\tilde{t}}{M^3}+\frac{2\lambda}{M^2}\leq \frac{2}{M^2},
\end{eqnarray*} where the last inequality holds because $n\log n=o(M)$ and $\tilde{t}=\Theta(n\log n).$ The inequality above implies that $|{\bf e}(t)|\leq \frac{2t}{M^2}$ and the lemma holds.
\end{proof}

\begin{lemma}
For any $t\geq \tilde{t},$ we have
\begin{align*}
|{\bf e}(t)|&\leq 4|{\bf e}(\tilde{t})|\exp\left(-\delta'(t-\tilde{t})\right)+ 128\lambda \epsilon^2 \exp\left(-\tilde{\delta}(t+\tilde{t})\right) \frac{1}{\tilde{\delta}}\left(1-\exp\left(-\tilde{\delta}(t-\tilde{t})\right)\right)\\
&=O\left(\frac{n\log n}{M^2}\right).
\end{align*}\label{lem:etlarget}
\end{lemma}
\begin{proof}
We now consider $t\geq \tilde{t}.$ Define Lyapunov function $$V({\bf e}(t))=\sum_{k=1}^n \tilde{w}_k |e_k(t)|$$ for $\tilde{w}_k$ defined previously. We first have
\begin{eqnarray*}
&&\dot{e}_k(t)\\
&=& f_k\left({\bf x}(t,\epsilon)\right)-f_k\left({\bf x}^{(0)}(t)\right)-\epsilon\frac{\partial f_k}{\partial x}({\bf x}^{(0)}(t)){\bf x}^{(1)}(t)\\
&=& f_k\left( {\bf x}^{(0)}(t)+\epsilon{\bf x}^{(1)}(t)+{\bf e}(t)\right)-f_k\left({\bf x}^{(0)}(t)\right)-\epsilon\frac{\partial f_k}{\partial x}({\bf x}^{(0)}(t)){\bf x}^{(1)}(t)\\
&=&\lambda\left(2\left(x_{k-1}^{(0)}+s_{k-1}^*+\frac{e_{k-1}}{2}+\epsilon x^{(1)}_{k-1}\right)e_{k-1}-2\left(x_{k}^{(0)}+s_{k}^*+\frac{e_k}{2}+\epsilon x^{(1)}_{k}\right)e_{k}\right)\\
&&-(e_{k}-e_{k+1}) +\left(\lambda \epsilon^2 \left(x^{(1)}_{k-1}\right)^2 -\lambda \epsilon^2 \left(x^{(1)}_{k}\right)^2\right).
\end{eqnarray*}

Define $$\dot{V}(t)=\sum_{k=1}^n W_k(t)+W(t),$$ where $W_k(t)$ includes all the terms involving $e_k(t)$ and $W(t)$ includes all the remaining terms. Note that $|e_k(t)|=O(n\log n/M^2)$ according to the previous lemma and $\epsilon |x_k^{(1)}|=O(1/M),$ both of which can be made arbitrarily small by choosing sufficiently large $M.$ Therefore, following the analysis of Lemma \ref{lem:x1}, we have
$$\sum_{k=1}^n W_k(t)\leq -\tilde{\delta} V(t),$$ which implies that
\begin{eqnarray*}
\dot{V}({\bf e}(t))\leq -\tilde{\delta} {V}({\bf e}(t))+8\lambda \epsilon^2 \sum_{k=1}^n \left(x^{(1)}_{k}\right)^2.
\end{eqnarray*}

Define $$A(t)=\sum_{k=1}^n \left(x^{(1)}_{k}\right)^2.$$
By the comparison lemma in  \cite{Kha_01}, we have
\begin{eqnarray}
V(t)&\leq&\phi(t-\tilde{t},0)V(\tilde{t})+8\lambda \epsilon^2 \int_{0}^{t-\tilde{t}} \phi(t-\tilde{t},\tau)A(\tau+\tilde{t})\,d\tau\label{eq:w}
\end{eqnarray} where the transition function $\phi(t,\tau)$ is
$$\phi(t,0)=\exp\left(-\delta't\right).$$

According to Corollary \ref{cor:x1bound}, we have
$$A(t)=\sum_{k=1}^n \left(x^{(1)}_{k}\right)^2\leq  \left(\sum_{k=1}^n \left|x^{(1)}_{k}\right|\right)^2\leq  16 e^{-2\tilde{\delta} t}.$$
Substituting the bounds on $\phi(t,\tau)$ and $A(\tau),$ we obtain
\begin{eqnarray*}
V(t)&\leq&  V(\tilde{t}) \exp\left(-\tilde{\delta}(t-\tilde{t})\right)+128\lambda \epsilon^2  \int_0^{t-\tilde{t}} \exp\left(-\tilde{\delta}(t-\tilde{t}-\tau)-2\tilde{\delta}(\tau+\tilde{t})\right)\,d\tau\\
&=& V(\tilde{t}) \exp\left(-\tilde{\delta}(t-\tilde{t})\right)+128\lambda \epsilon^2 \exp\left(-\tilde{\delta}(t+\tilde{t})\right) \int_0^{t-\tilde{t}} \exp\left(-\tilde{\delta}\tau\right)\,d\tau\\
&=& V(\tilde{t}) \exp\left(-\tilde{\delta}(t-\tilde{t})\right)+128\lambda \epsilon^2 \exp\left(-\tilde{\delta}(t+\tilde{t})\right) \frac{1}{\tilde{\delta}}\left(1-\exp\left(-\tilde{\delta}(t-\tilde{t})\right)\right).
\end{eqnarray*} The lemma holds due to the definition of $V(t)$ and the fact that $0\leq \tilde{w}_k\leq 4$ for any $k.$
\end{proof}

\begin{theorem}For sufficiently large $M,$ we have
\begin{eqnarray*}
\int_0^\infty |{\bf e}(t)|\,dt=O\left(\frac{\left(n\log n\right)^2}{M^2}\right).
\end{eqnarray*}\label{thm:errorbound}
\end{theorem}
\begin{proof}Combining the previous two lemmas, we obtain
\begin{eqnarray*}
\int_0^\infty |{\bf e}(t)|\,dt&=&\int_0^{\tilde{t}} |{\bf e}(t)|\,dt+\int_{\tilde{t}}^{\infty} |{\bf e}(t)|\,dt\\
&\leq& \frac{2\tilde{t}^2}{M^2} +\frac{4|{\bf e}(\tilde{t})|}{\tilde{\delta}}+\frac{128\lambda}{\tilde{\delta}^2}\epsilon^2\\
&\leq& \frac{2\tilde{t}^2}{M^2} +\frac{8\tilde{t}}{\tilde{\delta}M^2}+\frac{128\lambda}{\tilde{\delta}^2}\epsilon^2.
\end{eqnarray*}  Since $\tilde{t}=\Theta(n\log n)$ and $\tilde{\delta}=\Theta(n),$ the theorem holds.
\end{proof}

\section{Conclusions}
This paper proved that the stationary distribution of the power-of-two-choices converges in mean-square to its mean-field limit with rate $O\left(\frac{(\log M)^3(\log\log M)^2}{M}\right).$ The proof was based on Stein's method and the perturbation theory. The proof extended the result in \cite{Yin_16} to infinite-dimensional systems. Besides quantifying the rate of convergence for the power-of-two-choices, the approach based on truncated mean-field models has the potential to be applied to understand the rate of convergence other infinite-dimensional CTMCs to their mean-field limits.

\section*{Acknowledgement} The author thanks Prof. Jim Dai, Anton Braverman, and Dheeraj Narasimha for very helpful discussions. This work was supported in part by the NSF under Grant ECCS-1255425.

\bibliographystyle{IEEEtran}
\bibliography{U:/bib/inlab-refs}
\end{document}